\documentclass[final,pagebackref,letterpaper=true,colorlinks=true,pdfpagemode=none,urlcolor=blue,linkcolor=blue,citecolor=blue,pdfstartview=FitH]{siamltex}
\usepackage[parfill]{parskip}

\usepackage{paralist}
\usepackage{fancybox}
\usepackage{hyperref}
\usepackage{xcolor}
\usepackage{framed}
\usepackage{amsmath}
\usepackage{latexsym}
\usepackage{amssymb}

\newtheorem{claim}{Claim}[theorem]

\newtheorem{remark}{Remark}

\def\doctype{2}

\def\tsubmission{1}

\ifnum\doctype=\tsubmission
\newcommand{\full}[1]{}
\newcommand{\submit}[1]{#1}
\else
\newcommand{\full}[1]{#1}
\newcommand{\submit}[1]{}
\fi

\def\asflag{1}
\def\yes{1}

\ifnum \asflag=\yes
	\newcommand{\as}[1]{#1}
\else
	\newcommand{\as}[1]{}
\fi

\def\ceil#1{\lceil {#1} \rceil}

\newenvironment{proofof}[1]{\smallskip\noindent{\bf Proof of #1:}}%
        {\hspace*{\fill}$\Box$\par}

\newcommand{\ignore}[1]{}
\newcommand{\D}{\mathsf B}
\newcommand{\R}{\mathbb R}

\newcommand{\eps}{\varepsilon}

\newcommand{\poly}{\mathrm{poly}}

\newcommand{\md}[1]{\ (\operatorname{mod} #1)}

\newcommand{\cE}{{\mathcal E}}
\newcommand{\cP}{{\mathcal P}}

\newcommand{\tomega}{\widetilde{\Omega}}

\newcommand{\pathtester}{{\tt path-tester}}

\newcommand{\bS}{\mathbf{S}}

\newcommand{\Exp}{\mathbf{E}}
\newcommand{\Sec}[1]{\hyperref[sec:#1]{\S\ref*{sec:#1}}} %section
\newcommand{\Eqn}[1]{\hyperref[eq:#1]{(\ref*{eq:#1})}} %equation
\newcommand{\Fig}[1]{\hyperref[fig:#1]{Fig.\,\ref*{fig:#1}}} %figure
\newcommand{\Tab}[1]{\hyperref[tab:#1]{Tab.\,\ref*{tab:#1}}} %table
\newcommand{\Thm}[1]{\hyperref[thm:#1]{Theorem\,\ref*{thm:#1}}} %theorem
\newcommand{\Lem}[1]{\hyperref[lem:#1]{Lemma\,\ref*{lem:#1}}} %lemma
\newcommand{\Cor}[1]{\hyperref[cor:#1]{Corollary~\ref*{cor:#1}}} %corollary
\newcommand{\Def}[1]{\hyperref[def:#1]{Definition~\ref*{def:#1}}} %definition
\newcommand{\Alg}[1]{\hyperref[alg:#1]{Alg.~\ref*{alg:#1}}} %algorithm
\newcommand{\Ex}[1]{\hyperref[ex:#1]{Ex.~\ref*{ex:#1}}} %example
\newcommand{\Clm}[1]{\hyperref[clm:#1]{Claim~\ref*{clm:#1}}} %example
\newcommand{\Prop}[1]{\hyperref[prop:#1]{Prop.~\ref*{prop:#1}}} %property

\colorlet{shadecolor}{blue!10}
\begin{document}
%\full{\pagestyle{headings}}
\title{\sc A \lowercase{$o(n)$} Monotonicity Tester for Boolean Functions over the Hypercube}

\author{D.  Chakrabarty\thanks{Microsoft Research, 9 Lavelle Road, Bangalore 560001, India. {\tt dechakr@microsoft.com}}
\and
C. Seshadhri\thanks{Sandia National Labs, Livermore, USA. {\tt scomand@sandia.gov}}}

%\address{Microsoft Research India, 9 Lavelle Road, Bangalore, 560001}
%\email{dechakr@microsoft.com}
%\author[C. Seshadhri]{\sc C. Seshadhri}
%\address{Sandia National Labs, Livermore}
%\thanks{Sandia National Laboratories is a multi-program laboratory managed and operated by Sandia Corporation, a wholly owned subsidiary of Lockheed Martin Corporation, for the U.S. Department of Energy's National Nuclear Security Administration under contract DE-AC04-94AL85000.}
%
%\email{scomand@sandia.gov}

\def\hf{\hat{f}}
\def\hg{\hat{g}}
\def\Inf{{\sf Inf}}
\def\Viol{{\sf Viol}}
\def\I{{\mathbf I}}
\def\V{{\mathbf V}}
\def\R{{\mathbf R}}
%\full{\begin{nouppercase}
%\maketitle
%\end{nouppercase}
%
%}
%\submit{\maketitle}

\maketitle
\begin{abstract}
A Boolean function $f:\{0,1\}^n \mapsto \{0,1\}$
is said to be $\eps$-far from monotone if $f$ needs to be modified in at least $\eps$-fraction of the points to make it monotone. We design a randomized tester that is given oracle access to $f$ and an input parameter $\eps>0$, and has the following guarantee: It outputs {\sf Yes} if the function is monotonically non-decreasing, and outputs {\sf No} with probability $>2/3$, if the function is $\eps$-far from  monotone.  This non-adaptive, one-sided tester makes  \submit{\\}$O(n^{7/8}\eps^{-3/2}\ln(1/\eps))$ queries to the oracle.

%This answers a question by posed Goldreich et al.~\cite{GGLRS00}, who exhibited a $O(n/\eps)$ query tester, and asked whether testers exist with a significantly lower dependence on $n$. 
\end{abstract}
\submit{
%A category including the fourth, optional field follows...
}
\def\p{\mathbf p}
\def\B{\{0,1\}^n}

%\newpage
\section{Introduction}
Testing monotonicity of Boolean functions is a classical question in property testing. 
The Boolean hypercube $\B$ defines a  natural partial order with $x\prec y$ iff $x_i\leq y_i$ for all $i\in [n]$. A Boolean function $f:\{0,1\}^n\mapsto \{0,1\}$ is {\em monotone} if $f(x) \leq f(y)$ whenever $x\prec y$. 

A Boolean function's {\em distance} to monotonicity is the minimum fraction of points at which it needs to be modified to make it monotone. In the property testing framework we are provided oracle access to the function $f$ and are given a parameter $\eps > 0$. A \emph{monotonicity tester} is an algorithm that accepts if the function is monotone, and rejects if the function is $\eps$-far from monotone.  The tester is allowed to be randomized, and has to be correct with non-trivial probability (say $> 2/3$).   The tester is called {\em one-sided} if the tester always accepts a monotone function. The tester is {\em non-adaptive} if the queries made by the algorithm do not depend on the answers given by the oracle.

The quality of a monotonicity tester is governed by the number of oracle queries as well as the running time. Goldreich et al.~\cite{GGLRS00} suggested the following simple tester: query the function value on a pair of points that differ on exactly a single coordinate and reject if monotonicity is violated. In other words, the tester samples a random edge of the hypercube and checks for monotonicity between the two endpoints. This is called the {\em edge tester} for monotonicity. \full{It is clear the running time is of the same order as the number of queries}

Goldreich et al.~\cite{GGLRS00} show that $O(n/\eps)$-queries by the edge tester suffice to test monotonicity. Their analysis is tight, so the edge tester can do no better. They explicitly ask whether there exists a tester with an improved query complexity in terms of $n$. Fischer et al.~\cite{FLNRRS02} show that any non-adaptive, one-sided tester\footnote{\cite{FLNRRS02} also show a $\Omega(\log n)$ lower bound for $2$-sided testers.} for monotonicity must make 
$\Omega(\sqrt{n})$-queries for constant $\eps>0$. While monotonicity has been extensively studied in property testing~\cite{EKK+00, GGLRS00,DGLRRS99,LR01,FLNRRS02,HK04,PRR04,ACCL04,BCG+10,BBM11,CS12}, no significant progress had been made on this decade old question of testing monotonicity of Boolean functions. Our main result is an affirmative answer to the above question of~\cite{GGLRS00}.

%\submit{\vspace{-3mm}}
\begin{theorem}\label{thm:main}
There exists a one-sided, non-adaptive\submit{\\} $O(n^{7/8}\eps^{-3/2}\ln^{-1}(1/\eps))$-query monotonicity tester for Boolean functions $f:\{0,1\}^n \mapsto \{0,1\}$.
\end{theorem}
%\submit{\vspace{-3mm}}

\as{
We get an improved bound for functions with \emph{low average sensitivity}. Given a Boolean function $f$, the influence of dimension $i$ is the  fraction of edges of the hypercube crossing the $i$th dimension whose endpoints have different function values.  The average sensitivity, denoted as $\I(f)$, is the sum of all the $n$ influences. The functions defined in~\cite{FLNRRS02} to prove the lower bound of $\Omega(\sqrt{n})$ for non-adaptive, one-sided testers have constant average sensitivity, and hence the following is optimal for such functions.

\submit{\vspace{-3mm}}
\begin{theorem}\label{thm:asmon}
There exists a one-sided, non-adaptive\submit{\\} $O(n^{1/2}\eps^{-6}\I^{3}(f)\ln(1/\eps))$-query monotonicity tester 
for Boolean functions of average sensitivity $\I(f)$.
\end{theorem}
}

\begin{remark}[\bf Pair testers]
A {\em pair tester}~\cite{DGLRRS99}  describes a fixed distribution (independent of the function) on domain pairs $(x\prec y)$, makes independent queries on pairs drawn from this distribution, and rejects iff some drawn pair violates monotonicity.  By definition, pair testers are non-adaptive and one-sided. (Note that the edge tester is a pair tester.) 
Bri\"et et al.~\cite{BCG+10} show that any pair tester with a linear dependence of $\eps^{-1}$ must make  $\Omega(n/(\eps \log n))$ queries. The linear dependence is crucial in their argument. 
Our tester is also a pair tester. We circumvent the lower bound (on $n$) of~\cite{BCG+10} because of the worse dependence on $\eps$.
\end{remark}

\subsection{Main Ideas}\label{sec:main}
Our tester is a combination of the edge tester and what we call the {\em path} tester.
The path tester essentially does the following. It samples a random point $x$ on the hypercube, performs a sufficiently long random length walk on the {\em directed} hypercube to reach $y$, queries $f(x)$ and $f(y)$ and tests for monotonicity. We stress that the path tester does not query all the points along the path, but just the end points.

%As stated above, our algorithm samples a random path and picks two vertices randomly from it. 
Our algorithm is inspired by 
%The inspiration for our algorithm is 
a recent paper by Ron et al.~\cite{RRSW11}, which shows a $O(\sqrt{n})$-query randomized algorithm to estimate the average sensitivity of a {\em monotone} function. The algorithm essentially performs the operation above and counts
the number of mismatches;  Ron et al.~\cite{RRSW11} explicitly ask whether an algorithm ``in the spirit" above can be used for monotonicity.
Our answer is yes.
%Loosely, the algorithm described in  \cite{RRSW11} performs a random walk of around $\Theta(\sqrt{n})$ steps from a random point on the directed hypercube, and checks if the endpoints evaluate to different values.  Ron et al.~\cite{RRSW11} explicitly ask whether an algorithm ``in the spirit" above can be used for monotonicity.
%Our answer is yes.

Consider a function $f$ which is $\eps$-far from monotone.
The aim of any tester is to detect a \emph{violation}, that is, a pair $x \prec y$ such that $1 = f(x) > f(y) = 0$.
The success probability of the edge tester is exactly the fraction of violated edges. The intuition is that
there are possibly many more violations that are ``far away" and the directed random walk will help detect those.
%
%
%Why is performing a random walk and testing ``far away'' points is useful?
Consider the function $f:\{0,1\}^{n+1}\mapsto\{0,1\}$, $f(0,x) = 0$ if $|x|\leq n/2-2\sqrt{n}$ and $1$ otherwise; 
$f(1,x) = 0$ if $|x| \leq n/2 + 2\sqrt{n}$ and $1$ otherwise.  Here $|x|$ is the number of $1$s in $x$.
This function has a constant distance to monotonicity, and all the violated edges are of the form $((0,x),(1,x))$ for $n/2-2\sqrt{n} \leq |x|\leq n/2+2\sqrt{n}$. 
The edge tester detects a violation with probability only $\Theta(1/n)$. Suppose we pick a uniform random point and perform a random walk
of length $\sqrt{n}/2\leq \ell<\sqrt{n}$. If the starting point has $0$ in the first coordinate and any of the $\Theta(\sqrt{n})$ steps flips the first coordinate, 
the end points of the walk exhibit a violation to monotonicity. This happens with probability $\Theta(1/\sqrt{n})$, handily beating the edge tester.
%
%
%On the other hand, the random walk tester picks a point $(0,x), |x|=n/2\pm \sqrt{n}$ 
%with constant probability and performs a walk of length $\sqrt{n}/2\leq \ell<\sqrt{n}$. If any of these $\Theta(\sqrt{n})$ steps flip the first coordinate, 
%the walk detects a violation. Overall, a violation is detected with probability $\Theta(1/\sqrt{n})$, handily beating the edge tester.

The argument above required the violated edges to be aligned along one dimension.
In \Sec{match}, we prove that a directed random walk detects a violation with sufficiently high probability when there  is a large \emph{matching} of violated edges. One of the ingredients of this proof is the following  interesting combinatorial
observation. In \Sec{blue}, we prove that if a $\sigma$-fraction of the hypercube is marked blue, then the probability that the random walk starts and ends at a blue point is $\tomega(\sigma^2)$. It shows that the endpoints of this random walk, which are highly correlated, behave like two independent samples as far as being blue is concerned. 

But what if no large matching of violated edges exists? Take the `anti-majority' function defined as $f(x) = 1$ if $|x|\leq n/2$, and $f(x)=0$ otherwise. This function
is $1/2$-far from monotone, and yet the largest matching of violated edges is of size $\Theta(2^n/\sqrt{n})$. This is dealt with by our {\em dichotomy theorem}. In \Sec{dich}, we prove that for any $s>0$, either there exists a set of $\Theta(s\eps2^n)$ violated edges,
or there exists a matching of $\Theta(\eps2^n/s)$ violated edges. With this we are done; in the former case, the edge tester suffices, in the latter the path tester suffices. 
%The factor $n^{5/6}$ is the optimal tradeoff obtained by our approach.

The proof of our dichotomy theorem combines two ideas discovered earlier in the context of monotonicity testing. 
The first is a theorem of Lehman and Ron \cite{LR01} on multiple source-sink routing over the hypercube. 
The second is an alternating paths
machinery developed by the authors in a separate work~\cite{CS12}\submit{(an extended abstract appears in this proceedings)} on general range monotonicity testing.

\subsection{Isoperimetry for the directed hypercube} \label{sec:iso}

The problem of Boolean monotonicity testing is intimately connected with isoperimetric questions
on the \emph{directed} hypercube. We use $E$ for the set of undirected edges of the hypercube and $E(S,T)$
for the set of undirected edges from $S$ to $T$. Similarly, we use $E^+$ and $E^+(S,T)$ to 
denote the directed versions.

Any function $f:\{0,1\}^n \mapsto \{0,1\}$ can be thought
of as an indicator for the subset $S = \{x | f(x) = 1\}$. We use $\mu$ to denote
$|S|/2^n$, the uniform measure of $S$.
Let $\Phi(S)$ be the total influence of $S$,
which is $|E(S,\overline{S})|/2^{n-1}$. Let $\partial(S)$ be the \emph{boundary of $S$},
that is, $\{x | (x,y) \in E, x \in S, y \notin S\}$. The standard edge isoperimetric bound
for the undirected hypercube states that $\Phi(S) \geq 2\mu$, whenever $\mu \leq 1/2$. Harper's theorem~\cite{Har66} proves
that $|\partial(S)|$ is minimized when $S$ is a Hamming ball. Margulis~\cite{Mar74} proves
the remarkable fact that \emph{both} $\Phi(S)$ and $\partial(S)$ cannot be minimized simultaneously.
Formally, he proves that $\Phi(S)\cdot |\partial(S)| = \Omega(\mu^2)$, whenever $\mu\leq 1/2$.
(This is actually proven for the general  $p$-biased measures.)

What about the directed hypercube? We can define $\Phi^+(S) = |E^+(S,\overline{S})|/2^{n-1}$
and $\partial^+(S) = \{x | (x,y) \in E^+, x \in S, y \notin S\}$. Let $\eps_f$ denote
the distance of $f$ to monotonicity.
The success probability of the edge tester
is precisely $\Phi^+(S)/n$, and the classic theorem of Goldreich et al.~\cite{GGLRS00}
proves that $\Phi^+(S) = \Omega(\eps_f)$. Our dichotomy theorem is really a directed
version of Margulis' theorem. We prove that $\Phi^+(S)\cdot |\partial^+(S)| = \Omega(\eps^2_f)$.
It is interesting to note how $\eps_f$ takes the place of $\mu$ in the undirected bounds.

\section{The Tester and its Analysis}\label{sec:2}
We start by setting some notation. For binary vectors  $x, y\in\{0,1\}^n$, 
$|x|$ is the number of $1$'s in $x$ and $\|x-y\|_1$ is the $\ell_1$-distance  between $x$ and $y$.
The all zeros and all ones vectors are denoted $0^n$ and $1^n$, respectively.
The directed hypercube is the directed graph with vertex set $\{0,1\}^n$, and a directed edge from $x$ to $y$ if $x\prec y$ and $\|y-x\|_1=1$. 
Throughout the paper,  u.a.r. stands for `uniformly at random'.

Our tester is given as input a parameter $\eps > 0$ and query access to $f$.
The tester will accept if $f$ is monotone, and reject with probability $>2/3$ if $f$ is $\eps$-far from being monotone. 
We assume without loss of generality that $\eps \leq 1/2$ since any function can be made monotone by changing at most $1/2$ of its values.
Furthermore, we will assume $\eps \geq n^{-1/4}$, as conversely, \Thm{main} holds true by dint of the edge tester itself.

We set the following parameters. Let $C_\eps =\sqrt{10\ln(1/\eps)}$. By the assumptions on $\eps$, we get $2 < C_\eps \leq 2\sqrt{\ln n}$.
Let $\ell := 2\ceil{C_\eps\sqrt{n}~}$. Note that $\ell > 4\sqrt{n}$, and $\frac{\ell}{C_\eps} = \Theta(\sqrt{n})$.
We let $I_\ell$ denote the index set $[n/2-\ell/2, n/2 + \ell/2]$. 
For $1\leq i\leq n$, $L_i := \{x\in \{0,1\}^n: |x|=i\}$ denotes the $i$th layer of the directed hypercube.
We refer to $\bigcup_{i \in I_\ell} L_i$ as the \emph{middle layers} of the hypercube. We say an edge of the hypercube lies in the middle layers if both its endpoints lie in the middle layers.

We now describe the random walk based procedure called the \emph{path tester}.
This uses a parameter $\sigma \in (0,1)$ that determines the distance between samples.
%This 
%
%Our final tester is runs with probability $1/2$ we run one or the other.
%
%The first is the edge tester described in the introduction: pick a random edge of the hypercube, query the function at the endpoints, and reject if there is a violation of monotonicity among these. 
%The second is the path tester which is described below.
\medskip
\begin{shaded}
%\noindent
{\bf \pathtester$(\sigma)$.}
\submit{\vspace{-3mm}}
\begin{enumerate}
\item Let $\cP$ be the collection of paths in the directed hypercube from $0^n$ to $1^n$.
Pick a path $\p\in \cP$ u.a.r.  Let $X_\p:=\{z\in \p: |z|\in I_\ell\}$.
\item Sample $x\in X_\p$ u.a.r.
\item Let $Y_\p(x) := \{z\in X_\p: ||z-x||_1 \geq \frac{\sigma\ell}{32C_\eps} - 1\}$.
Sample $y\in Y_\p(x)$ u.a.r.
% which are at a distance $\geq \floor{\frac{\eps\ell}{32C_\eps}}$ from $x$ on $\p$.
\item Reject if $(x,y)$ violates monotonicity; i.e. $f(x) < f(y), \ x\succ y$ or $f(x) > f(y), \ x\prec y$.
\end{enumerate}
\end{shaded}

%Simply put, the path tester samples a path from $0^n$ , samples a random point in the `middle layers', and then samples another point `sufficiently' far away from the first point, queries the function on this pair and tests for monotonicity.
%\begin{theorem}%\label{thm:main}
%There is an $O(n^{19/20})$-query monotonicity tester for functions $f:\{0,1\}^n \mapsto \{0,1\}$.
%\end{theorem}

This is clearly a pair tester (and is hence non-adaptive and one-sided). 
% It is \emph{not} proximity oblivious~\cite{GoRo09}, since $Y_\p(x)$ depends on $\eps$. 
Our final tester runs either the path tester with a particular $\sigma$ to be fixed later, or the edge tester, each with probability $1/2$.
%
%%It suffices to show that for functions which are $\eps$-far from monotone, we catch a violation with probability at least $\Omega(\frac{\poly(\eps)}{n^{19/20}})$. Our starting point are maximal matchings $M$ of violating pairs. We choose $M$, as we did in our earlier paper~\cite{CS12}. We state the relevant properties when needed; for now, it is enough to recollect that $|M|\geq \eps 2^{n-1}$.
%%We now prove \Thm{main}. 
%Our tester is clearly non-adaptive.
%It is one-sided because both the edge and the path tester always accept if the function $f$ is monotone. 

The challenge lies in lower bounding the probability of rejection when the function $f$ is $\eps$-far from monotonicity.
%In the remainder of this section, we
%show that if $f$ is $\eps$-far from monotone, then the probability of rejection  is at least $\Omega(n^{-5/6}\eps^{5/3})$.  
%This will complete the proof.
%
%We do not try to optimize the exact exponent here, since our aim was to show it is less than $1$. This is 
%done for a cleaner presentation. (As we explain later on, our argument at best can give a running time of $O(n^{0.833})$, and
%there is no reason to believe this is optimal. So we side with a cleaner presentation over a smaller exponent.)
%
Henceforth, we assume the function $f$ is $\eps$-far, and we call the rejection event a success. 
Since $f$ is $\eps$-far from monotonicity, any maximal set $M$ of disjoint, violating pairs satisfies 
$|M|\geq \eps2^{n-1}$ (Lemma 3 of~\cite{FLNRRS02}). We refer to $M$ as a {\em matching} of violated pairs.

We start with an easy proposition. 
\begin{proposition} \label{prop:hyper} 
{\em (a)} $|\bigcup_{i \notin I_\ell} L_i| \leq \eps^5 2^n$.
{\em (b)} For all $i\leq n$, a u.a.r path $\p$ contains a u.a.r vertex from $L_i$.
\end{proposition}
%\begin{asparaitem}
%	\item $|\bigcup_{i \notin I_\ell} L_i| \leq \eps^5 2^n$.
%	\item For all $i$, a u.a.r. $\p$ contains a uniform random vertex from $L_i$.
%\end{asparaitem}
%\end{proposition}
\begin{proof} %The first is just a tail bound for binomial coefficients and follows from a standard Chernoff bound (Theorem 1.1 in \cite{DuPa09}).
%If we flip $n$ independent unbiased coins and denote the number of heads by $X$, then 
By Chernoff bounds, for a u.a.r $x\in \{0,1\}^n$, $\Pr[\big||x| - n/2\big| > \ell/2] \leq 2e^{-\ell^2/2n}$.
Since $\ell = 2\ceil{\sqrt{10n\ln(1/\eps)}~}$, this probability is at most $\eps^5$.
For the second part, observe that the number of paths in $\cP$ that pass through a given vertex $x$ depends solely on $|x|$.
\end{proof}

%\comment{dc: add outline}

\subsection{Going from blue to blue} \label{sec:blue} 
Suppose at least $\sigma 2^n$ vertices of the middle layers are colored blue. 
%(Note that this is the same $\sigma$ that is a parameter for the path tester.)
Let $(x,y)$ be a random pair sampled by \pathtester$(\sigma)$, and let $\cE$ be the event that both $x$ and $y$ are blue.
If $x$ and $y$ were chosen {\em independently} u.a.r., then the probability
of both being blue is $\sigma^2$. 
The following lemma shows that this probability does not degrade much even though $x$ and $y$ are correlated (for instance, they form an ancestor-descendant pair).
%The following lemma lower bounds the event probability.

%Let $\mu > 0$ be a parameter.
%We select a path $\p$ u.a.r. from $\cP$. We sample $x$ u.a.r. from $X_\p$, and $y$ u.a.r. from
%the set $Y'_\p(x) = \{z \in X_\p: \|x - z\|_1 \geq \mu \ell/2 - 1\}$. 
%

\begin{lemma} \label{lem:blue} 
%Let $\cE$ denote the event that $x$ and $y$ are both blue, where $x$ and $y$ are a pair chosen
%by the path tester. Then 
$\Pr[\cE] = \Omega\left(\dfrac{\sigma^2}{\ln(1/\eps)}\right)$.
\end{lemma}

\begin{proof}
For notational convenience, set $\mu := \sigma/16C_\eps$.
This implies\footnote{The curious reader may be wonder why we have a ``$-1$" in the distance condition for $Y_\p(x)$ in the description of the path tester.
This is a technicality so that we have the bound $|X_\p|-|Y_\p(x)| \leq \mu \ell$. Without the $-1$,
the bound would be $\mu \ell + 2$ and can be made $O(\mu \ell)$ only for large
enough $\sigma$. So instead of enforcing such a condition or carrying around a $+2$,
the ``$-1$'' allows for a cleaner presentation.}  $|X_\p| - |Y_\p(x)| \leq \mu \ell$ for any $x\in \p$.
%Let $B(\p)$ be the set of blue points in $X_\p$ corresponding to a path $\p$, and $b(\p) := |B(\p)|$ be the random variable denoting the number of blue points. 
Let $b(\p)$ be the random variable denoting the number of blue points in $X_\p$ corresponding to a random path $\p$.
Let $\cE_x$ and $\cE_y$ be the events that the first and second points are blue; that is $\cE = \cE_x \wedge \cE_y$.
Abusing notation, $\p$ will also denote the event that $\p$ is the sampled path.

Conditioned on a path $\p$ being sampled, the probability of the first point $x$ sampled by the path tester being blue is $b(\p)/\ell$. Formally, $\Pr[\cE_x~|~\p ]  = b(\p)/\ell$.
% $$\Pr[\cE_x~|~\p \textrm{ sampled }] = \frac{b(\p)}{\ell}$$

Conditioned on the path being $\p$ and the first point being $x$ (irrespective of it being blue or not), the probability that the second point $y$ is blue is the number of blue points in $Y_\p(x)$ divided by 
$|Y_\p(x)|$. The number of blue points in $Y_\p(x)$ is at least $b(\p) - \mu\ell$ since $|X_\p| - |Y_\p(x)| \leq \mu \ell$.
Therefore,
$$\Pr[\cE_y ~|~ \p,x \textrm{ first point}] =  \frac{|\textrm {blue points in $Y_\p(x)$}|}{|Y_\p(x)|}\geq \frac{b(\p) - \mu\ell}{\ell}.$$

Since the above inequality holds for all $x$ (in particular, any blue $x$), 
$$\Pr[\cE_y~|~\p,\cE_x] \geq \frac{b(\p) - \mu\ell}{\ell}.$$

\ignore{
Now, 
$$\Pr[\cE_y|\p,\cE_x] = \frac{\sum_{x\in B(\p)} \Pr[\cE_y|\p,x \textrm{ first point}]\Pr[x \textrm{ first point}] }{\sum_{x\in B(\p)} \Pr[x \textrm{ first point }]}$$
 }
%This is because there are at least $(b(\p) - \eps\ell/32C_\eps)$ blue points in $Y_\p(x)$. 
Together, we get%The probability that both points sampled by the random path tester are blue is 
\begin{eqnarray}
\Pr[\cE]  & = & \sum_{\p\in\cP} \Pr[\cE_y|\p,\cE_x]\cdot\Pr[\cE_x|\p]\cdot\Pr[\p] \nonumber \\
& \geq &  \sum_{\p\in \cP} \left(\frac{b(\p) - \mu\ell}{\ell}\cdot\frac{b(\p)}{\ell}\cdot\frac{1}{|\cP|}\right) \nonumber \\
& = & \frac{1}{|\cP|}\sum_{\p\in \cP} \left(\frac{b(\p)}{\ell}\right)^2 - \frac{\mu}{|\cP|} \sum_{\p \in \cP} \frac{b(\p)}{\ell}. \label{eq:1}
\end{eqnarray}

In the following, we use $\Exp[...]$ to denote the expectation over the choice of the path $\p$. Note that $\Exp[b(\p)/\ell] := \frac{1}{|\cP|} \sum_{\p\in \cP} (b(\p)/\ell)$, anf thus we can express the bound of \Eqn{1} in terms of expectations.  The second inequality is an application of Jensen's inequality.
\begin{eqnarray}
	\Pr[\cE] & \geq & \Exp[(b(\p)/\ell)^2] - \mu \Exp[b(\p)/\ell] \notag \\
	& \geq & (\Exp[b(\p)/\ell])^2 - \mu \Exp[b(\p)/\ell] = \Exp[b(\p)/\ell] (\Exp[b(\p)/\ell] - \mu) \label{eq:newone}
\end{eqnarray}

The following claim lower bounds the expectation
\begin{claim} \label{clm:exp} $\Exp[b(\p)/\ell] \geq \frac{\sigma}{4C_\eps}$.
\end{claim}
\begin{proof}
Note that, for all $i$,
$|L_i|\leq {n\choose n/2} \leq \frac{2^n}{\sqrt{n}}$.
Let $n_i$ be the number of blue vertices in layer $L_i$. Note that $\sum_{i\in I_\ell} n_i\geq \sigma 2^n$. Let $Z_i$ be the indicator variable for the $i$th layer vertex in $\p$ being blue.
Hence, $b(\p) = \sum_{i \in I_\ell} Z_i$.
For all $i$, a $\p$ chosen u.a.r from $\cP$ contains a uniform random vertex in layer $L_i$ (\Prop{hyper}).
Thus,
$$\Exp[Z_i] = \frac{n_i}{|L_i|} \geq \frac{\sqrt{n}}{2^n} \cdot n_i.$$
%
%By the tail bound of \Prop{hyper}, the number of vertices not in the middle layers is at most $\eps^52^n$.
%Since $\eps \leq 1/2$, $\sum_{i \in I_\ell} n_i \geq (\eps - 2\eps^5)2^n \geq \eps 2^{n-1}$.
%
%Furthermore, by a standard application of Chernoff bounds, we get that the number of points of the hypercube not lying in 
%any of the $L_i$'s is at most $2\eps^{5}\cdot 2^n$. 
%since $\eps \leq 1/2$. 
Using linearity of expectation and the bound $\ell < 4C_\eps \sqrt{n}$,
%For now, Let us assume $|L_i|$'s are the same for all $i$ and equals $2^n/\ell$; this is not true and will lead to ugly looking expressions, but it's not too inaccurate so Let us stick with this. Modulo this assumption, we get 
$$\Exp[b(\p)/\ell] \geq \frac{\sqrt{n}}{\ell 2^n}\sum_{i \in I_\ell} n_i \geq  \frac{\sigma\sqrt{n}}{\ell} \geq \frac{\sigma}{4C_\eps}.$$
%since $4\sqrt{n} < \ell = 2\ceil{C_\eps\sqrt{n}} \leq 4C_\eps\sqrt{n}$.
%since $
\end{proof}

%Let us use $\cQ := \{\p: b(\p) \geq \eps\ell/2\}$. 
%\noindent
%From \Eqn{expbp}, we get
%\begin{equation}\label{eq:q}
%\Pr[ \p \in \cQ] := \Pr[b(\p) \geq \eps\ell/8C_\eps] \geq \frac{\eps}{8C_\eps}
%\end{equation}
%This is because the maximum value of $b(\p)$ is $\ell$. If \Eqn{q} did not hold we would get,
%$$\Exp[b(\p)] \leq \ell\cdot \Pr[b(\p) \geq \eps\ell/8C_\eps] + \ell\eps/8C_\eps\cdot 1 < \ell\eps/4C_\eps$$ contradicting \Eqn{expbp}.
%
%Now we are ready lower bound $\Pr[\cE]$.  For convenience, we use $\tomega(\cdot)$ notation to hide log factors. 
%For a function $h < 1$ (resp. $h > 1$), we use $\tomega(h)$ for functions
%dominating $h/\poly(\log (1/h))$ (resp. $h/\poly(\log h)$).
%This immidiately gives us
%\begin{equation}\label{eq:expq}
%\frac{1}{|\cP|}\cdot \sum_{\p\in\cQ} \frac{b(\p)}{\ell} \geq \frac{\eps^2}{4}.
%\end{equation}
%We are almost there.
%What we want to lower bound is $\frac{1}{|\cP|} \sum_{\p\in \cQ}\left(b(\p)/\ell\right)^2$; we have the 
%\begin{lemma}\label{lem:piece1}
%$$\Pr[\cE] ~\geq~  \frac{1}{2|\cP|} \sum_{\p\in \cQ} \left(\frac{b(\p)}{\ell}\right)^2 ~\geq~ \frac{\eps^3}{1024~C^3_\eps} ~~= \tomega(\eps^{3}).$$
%\end{lemma}
%\begin{proof} (of \Lem{blue})
%
The function $h(x) = x(x-\mu)$ is increasing when $x \geq \mu/2$. 
The lower bound of \Clm{exp} gives $\Exp[b(\p)/\ell] \geq \sigma/4C_\eps > \mu$.
Substituting in \eqref{eq:newone} gives $\Pr[\cE] \geq (\sigma/4C_\eps)(\sigma/4C_\eps - \mu)$. Plugging back $\mu = \sigma/16C_\eps$ completes
the proof of \Lem{blue}.
%
%The first inequality is the same as \Eqn{2}.
%Note that this can be restated as 
%\begin{equation}\label{eq:l0}
%\Pr[\cE] \geq 
%\frac{1}{2}\frac{|\cQ|}{|\cP|}\cdot \frac{1}{|\cQ|}\sum_{\p\in \cQ} \left(\frac{b(\p)}{\ell}\right)^2.
%\end{equation}
%From \Eqn{q}, we get that 
%\begin{equation}\label{eq:l1}
%\frac{|\cQ|}{|\cP|} \geq \frac{\eps}{8C_\eps}.
%\end{equation}
%By Jensen's inequality, we get that 
%\begin{equation}\label{eq:l2}
%\frac{1}{|\cQ|}\sum_{\p\in \cQ} \left(\frac{b(\p)}{\ell}\right)^2 \geq \left(\frac{1}{|\cQ|}\sum_{\p\in \cQ}\frac{b(\p)}{\ell}\right)^2.
%\end{equation}
%Since for each $\p\in\cQ$, $b(\p)\geq \eps\ell/8C_\eps$, we get that
%\begin{equation}\label{eq:l3}
%\left(\frac{1}{|\cQ|}\sum_{\p\in \cQ}\frac{b(\p)}{\ell}\right)^2 \geq \frac{\eps^2}{64C^2_\eps}.
%\end{equation}
%Plugging \Eqn{l1},\Eqn{l2},\Eqn{l3} in \Eqn{l0}
%gives the lemma.
%%, and the rather weak inequality, $\ln(1/\eps)\leq 1/\eps$.
\end{proof}

\subsection{Large violated-edge matchings are good} \label{sec:match} 
We bound the success of the path tester when a large matching of violated 
edges exists. 
%For convenience, we use $\tomega(\cdot)$ notation to hide polylog factors. 
%For a function $h < 1$ (resp. $h > 1$), we use $\tomega(h)$ for functions
%dominating $h/\poly(\log (1/h))$ (resp. $h/\poly(\log h)$).

%We show that if a large matching $E$ of 
%violating edges exists, then the path tester succeeds with large probability.
%We assume that all edges of $E$ lie between the layers $(n/2-C_\eps\sqrt{n}, n/2+C_\eps\sqrt{n})$.
%That is, for all $(x,y)\in E$, we have $\frac{n}{2} - C_\eps\sqrt{n} < |x|,|y| < \frac{n}{2}+C_\eps\sqrt{n}$.
%We call such edges to be in the {\em middle layer}.
% $\geq \Omega\left(\frac{1}{\sqrt{n}}\left(|E|/2^n\right)^3\right)$.

\begin{lemma}\label{lem:piece2}
Suppose there exists a matching $E$ of violated edges all lying in the middle layers of the hypercube. 
Set $\sigma = |E|/2^n$.
Then \pathtester$(\sigma)$ succeeds with probability 
$ \Omega\left(\frac{\sigma^3}{\sqrt{n}\ln(1/\eps)}\right)$.
\end{lemma}

\begin{proof} 
We begin with some notation.
Let the set of endpoints of edges in $E$ be $B$. We partition $B$ into $B_0$ and $B_1$, indexed by the value of the function on these vertices.  
That is, $B_0 = \{x\in B: f(x) =0\}$ and $B_1 = \{x\in B: f(x)=1\}$.
Note that $|B_0| = |B_1| = |E|$. %Let $\eta$ denote the fraction $|E|/2^n$.
For any two points $x,y$, let $\cE_{x,y}$ denote the event that the path tester picks $(x,y)$. 
For convenience, in what follows, the pairs of $E$ will be ordered
according to the directed hypercube (so if $(z,z') \in E$, then $z \prec z'$).
We abuse notation to define $E$ as a function.
That is, for edge $(z,z') \in E$, we set $E(z) = z'$ and $E(z') = z$.

We define the following sets of pairs of vertices.
\begin{align*}
\Pi & = \{(x,y) | x \prec y, \|x-y\|_1 \geq \frac{\sigma\ell}{32C_\eps} - 1, x \in B_1, y \in B_1\}. \\
\Pi' & = \{(x,E(y)) | (x,y) \in \Pi\}.
\end{align*}
A few observations. $\Pi$ lies in the support of the pair tester, that is, pairs $(x,y)$ sampled with non-zero probability.
Every pair in $\Pi'$ is a violation; for $(x,y)\in \Pi$, we have $x\prec y$, $y\in B_1$ implying $E(y) \succ y$ and $E(y)\in B_0$. %since $x \in B_1$, $y \in B_0$, and $x \prec y$.
%Also, $\Pr[\cE_{x,y}]$ is non-zero for all pairs in $\Pi$.
%Crucially, we can define a 1-1 mapping from pairs in $\Pi'$ to those in $\Pi$.
%Finally, there is a onto mapping $\phi$ from pairs in $\Pi$ to pairs in $\Pi'$ given by 
%$\phi(x,y) = (x,E(y))$.
Finally, the mapping $(x,y)\in \Pi$ to $(x,E(y))\in \Pi'$ is one-to-one.
%If we map $(x,y') \in \Pi'$ to $(x,E(y'))$ (this is in $\Pi$), this is a 1-1 mapping from pairs in $\Pi'$ to those in $\Pi$.
This uses the fact that $E$ is a matching (and is a crucial piece of the proof).

%If the path tester pick $x \in B_1$ and an ancestor $y \in B_0$, then the tester is successful.

Since all pairs in $\Pi'$ are violations, $$\Pr[\textrm{ success }] \geq \sum_{(x,y')\in \Pi'} \Pr[\cE_{x,y'}].$$
Using the mapping between $\Pi'$ and $\Pi$ and that $\Pr[\cE_{x,y}] > 0$ for $(x,y) \in \Pi$,
\begin{eqnarray}
\sum_{(x,y')\in \Pi'} \Pr[\cE_{x,y'}] & = & \sum_{(x,y')\in \Pi'} \Pr[\cE_{x,E(y')}]\cdot \frac{\Pr[\cE_{x,y'}]}{\Pr[\cE_{x,E(y')}]} \nonumber \\
& = & \sum_{(x,y) \in \Pi} \Pr[\cE_{x,y}]\cdot \frac{\Pr[\cE_{x,E(y)}]}{\Pr[\cE_{x,y}]} \label{eq:mat}
\end{eqnarray}

We break the remaining proof into simpler claims. 
For a vertex $x$, define $s(x) := |Y_\p(x)| = |\{z \in X_\p: \|z - x\|_1 \geq \sigma\ell/32C_\eps - 1\}|$
where $\p$ is some path containing $x$. This is well-defined since $|Y_\p(x)|$ is {\em independent} of $\p$ for any $\p \ni x$. In fact, 
\begin{equation}\label{eq:defs}
s(x) = \Big|\left\{i\in I_\ell: \big|i - |x|\big|\geq \frac{\sigma\ell}{32C_\eps} -1\right\}\Big|
\end{equation}
The following claim is a routine calculation.

\begin{claim} \label{clm:euv} Suppose $x,y$ are in the middle layers and $\|x-y\|_1 \geq \frac{\sigma\ell}{32C_\eps} - 1$. 
Let $\cP_{x,y}$ denote the set of paths containing both $x$ and $y$.
Define
\begin{eqnarray*}
\theta_{x,y} := 
%\begin{cases}
%0 & \text{ if $||v-u||_1 < \frac{\sigma\ell}{32C_\eps} - 1$} \\
\frac{1}{\ell}\left( \frac{1}{s(x)} + \frac{1}{s(y)}\right) 
%& \text{ otherwise.}
%\end{cases} 
\end{eqnarray*}
Then,
\begin{equation}\label{eq:pxy}
\Pr[\cE_{x,y}] = \theta_{x,y} \frac{|\cP_{x,y}|}{|\cP|}
\end{equation}
\end{claim}

\begin{proof} 
%Think of $\cE_{u,v}$ as picking a path containing both $u$ and $v$, and conditioned on this, sampling $u$ and $v$. In other words, 
%
Note that
$$ \Pr[\cE_{x,y}] = \sum_{\p\in\cP_{x,y}} \Pr[\p \textrm{ sampled }]\cdot \Pr[x,y \textrm{ sampled }|~\p \textrm{ sampled }].$$
Since $\|x-y\|_1 \geq \frac{\sigma\ell}{32C_\eps} - 1$, $y \in Y_\p(x)$ (and vice versa).
Suppose $x$ is the first point to be sampled; this happens with probability $1/\ell$. 
%If $v\notin Y_p(u)$, then the probability $v$ is the second point sampled is $0$. Observe that if $v\notin Y_\p(u)$, then $u\notin Y_\p(v)$ as well. 
The probability that $y$ is the second point sampled is $\frac{1}{|Y_\p(x)|}$. 
Arguing analogously when $y$ is sampled first, when $x,y\in X_\p$, 
\begin{eqnarray*}
\Pr[x,y \textrm{ sampled }|~\p \textrm{ sampled }]  = 
%\begin{cases}
%0 & \text{ if $||v-u||_1 < \floor{\frac{\sigma\ell}{32C_\eps}}$} \\
\frac{1}{\ell}\left( \frac{1}{|Y_\p(x)|} + \frac{1}{|Y_\p(y)|}\right) = \theta_{x,y}.
%& \text{ otherwise.}
%\end{cases} 
\end{eqnarray*}

The proof concludes by noting that $\sum_{\p: x,y\in \p} \Pr[\p \textrm{ sampled}] = \frac{|\cP_{x,y}|}{|\cP|}$.
%Note that $|Y_\p(u)|$ is {\em independent} of $\p$ and depends only on $|u|$. 
%The probability that $u$ and $v$ are sampled conditioned on the path $\p$ is the {\em same} 
%irrespective of $\p$. It depends only on $u$ and $v$, as long as $\p$ contains both of them. 
%Therefore, 
%%If $x$ and $y$ are less than $\floor{\frac{\sigma\ell}{32C_\eps}}$ apart, then the probability is $0$.
%%Furthermore, if 
%%for all $x,y,\p$: if $\p$ does not contain both $x$ and $y$ it's zero, if $x$ and $y$ are within distance $\sigma\ell/8$ it's $0$, otherwise it does not matter.
%%Otherwise, it's precisely $\frac{1}{\ell}\cdot \frac{1}{\ell'}$, where $\ell'$ is the number of points at a distance $\geq \floor{\frac{\sigma\ell}{32C_\eps}}$ from $x$. Let us call this non-zero probability $\theta_{xy}$. 
%\begin{equation*}
%\Pr[\cE_{u,v}] = \theta_{u,v} \sum_{\p: u,v \in \p} \Pr[\p \textrm{ sampled }] \frac{|\cP_{u,v}|}{|\cP|}.
%\end{equation*}
%%where $\cP_{xy}$ are paths from $0^n$ to $1^n$ containing both $x$ and $y$.
\end{proof}

%
%Let us focus on pairs $(x,y)\in B_1\times B_1$. Let $\cE_{xy}$ denote the event that the path tester chooses $x$ and $y$. If we think of the vertices in $B_1$ to be colored blue, then \Lem{piece1} implies,
%\begin{equation}\label{eq:5}
%\sum_{(x,y)\in B_1\times B_1} \Pr[\cE_{xy}] ~ =\tomega(\eta^{3}).%\geq K\eta^{4.5}.
%\end{equation}
%
%
%Till now, we have been looking at the probability of picking two vertices in $B_1$. What we really need to sample is a vertex $x$ in $B_1$ and ancestor $y'$ in $B_0$. Here's where the matching will help us. 
%We can map every event $\cE_{xy'}$ to the event $\cE_{xy}$, where $y = E(y')$, the matched pair of $y'$. 
%Since we have a matching, the mapping is one-to-one.  Now, we can lower bound our success probability as 
%\begin{eqnarray*}
%\Pr[\textrm{ success }] & \geq & \sum_{(x,y')\in B_1\times B_0} \Pr[\cE_{xy'}] \\
%							     &  = & \sum_{(x,y)\in B_1\times B_1} \Pr[\cE_{xy}]\cdot \frac{\Pr[\cE_{xy'}]}{\Pr[\cE_{xy}]}
%\end{eqnarray*}
%%for which $\Pr[\cE_{xy}]> 0$ (and assuming $x\prec y$), we map the event $\cE_{xy}$ to $\cE_{xy'}$ where $(y,y')$ is the matching edge in $E$. Note, $y' \succ y$ and $y'\in B_0$. In particular, $x$ and $y'$ are $\ell\eps/32C_\eps$ apart. From \Eqn{pxy}, we get 

The next claim shows that for any $(x,y)\in \Pi$, $\theta_{x,E(y)}$ is almost as large as $\theta_{x,y}$.
\begin{claim} \label{clm:theta}
%For $(x,y) \in \Pi$, $\theta_{x,E(y)} \geq \left(1-\frac{1}{\sqrt{n}}\right)\theta_{x,y}$
For $(x,y) \in \Pi$, $\theta_{x,E(y)} \geq \theta_{x,y}/2$
\end{claim}
\begin{proof} 
\ignore{Since $(x,y) \in \Pi$, we have $x \prec y \prec E(y)$ and $\|x-y\|_1 \geq \frac{\sigma\ell}{32C_\eps} - 1$.
Note that $x$ is contained in the middle layers since the matching $E$ lies in the middle layers. Consider some path $\p$ passing through $x$, $y$, and $E(y)$.
Look at the sets $Y_\p(y) = \{z \in X_\p: \|z - y\|_1 \geq \frac{\sigma\ell}{32C_\eps} - 1\}$
and $Y_\p(E(y)) = \{z \in X_\p: \|z - E(y)\|_1 \geq \frac{\sigma\ell}{32C_\eps} - 1\}$.
Since $E(y) \succ y$ and $x \in X_\p$, we can conclude that $|Y_\p(E(y))| \geq |Y_\p(y)|$.
(Equivalently $S(E(y))^{-1} \leq S(y)^{-1}$.)
}
For convenience, let $y'$ denote $E(y)$; note that $y'\succ y$ and $|y'|=|y|+1$.
There exists some path containing $x, y$, and $y'$.
%an $x\prec y\prec y'$ such that $x\in X_\p$ for some path $\p$, we get 
From \Eqn{defs}, $s(y) \leq s(y') \leq s(y)+1$.

%We have $|Y_\p(E(y))| \geq |Y_\p(y)|$.
%We set $y' = E(y)$.
%Since $|y'| = |y|+1$, $||y'-x||_1 > ||y-x||_1$. Hence, sif $\theta_{x,y'} = 0$, so is $\theta_{x,y}$.
%Furthermore, since $|y'|<n/2-C_\eps\sqrt{n}$, for all paths $\p$ containing both $x$ and $y'$, we have $y'\in X_\p$. 
%
%Note that $S(E(y)) \leq S(y) + 1$ since the number of $1$'s in $y$ and $E(y)$ differ by at most $1$.
Putting it all together,
$$\frac{\theta_{x,y'}}{\theta_{x,y}} = \frac{s(x)^{-1} + s(y')^{-1}}{s(x)^{-1} + s(y)^{-1}}
\geq \frac{s(y')^{-1}}{s(y)^{-1}} \geq \frac{s(y)}{s(y)+1} \geq 1/2.$$
The first inequality follows from the observation $\frac{c+a}{c+b} \geq \frac{a}{b}$ whenever $a\leq b$ and $c\geq 0$.
Since $\ell = 2\ceil{C_\eps\sqrt{n}}$, $s(y) \geq \ell - \frac{\sigma\ell}{16C_\eps} \geq 1$,
yielding the final inequality
% since $\eps < 1/2$ and $\ell = 2C_\eps\sqrt{n}, C_\eps = \sqrt{10\ln (1/\eps)}$.
\end{proof}

\begin{claim} For $(x,y) \in \Pi$, $$\frac{\Pr[\cE_{x,E(y)}]}{\Pr[\cE_{x,y}]} = \Omega\left(\frac{\sigma}{\sqrt{n}}\right).$$
\end{claim}

\begin{proof} Combining \Clm{euv} and \Clm{theta},
\begin{equation}\label{eq:7}
\frac{\Pr[\cE_{x,E(y)}]}{\Pr[\cE_{x,y}]} = \frac{\theta_{x,E(y)}|\cP_{x,E(y)}|}{\theta_{x,y}|\cP_{x,y}|}
%\geq \left(1-\frac{1}{\sqrt{n}}\right)  \frac{|\cP_{x,E(y)}|}{|\cP_{x,y}|}
\geq \frac{|\cP_{x,E(y)}|}{2|\cP_{x,y}|}
\end{equation}

We know exactly what both the numbers in the RHS are. Say $|x|=t$ and $|y|=t+u$. Note $u\geq \sigma\ell/32C_\eps - 1$ and $|E(y)| = |y|+1$. Then, 
$$|\cP_{x,y}| = t!u!(n-u-t)! ~~~~ \textrm{ and } ~~~~ |\cP_{x,E(y)}| = t!(u+1)!(n-u-t-1)!$$
Plugging in \Eqn{7}, 
%$$\frac{\Pr[\cE_{x,E(y)}]}{\Pr[\cE_{x,y}]} \geq \left(1-\frac{1}{\sqrt{n}}\right)\frac{u+1}{n-u-t}.$$ 
$$\frac{\Pr[\cE_{x,E(y)}]}{\Pr[\cE_{x,y}]} \geq \frac{u+1}{2(n-u-t)}.$$ 
The denominator is $\Theta(n)$ since $n/2-C_\eps\sqrt{n}\leq|y|\leq n/2+C_\eps\sqrt{n}$, and $C_\eps \leq 2\sqrt{\ln n}$. The numerator is at least $\sigma\ell/32C_\eps = \Omega(\sigma\sqrt{n})$,
completing the proof.
% Thus we get 
%$$\Pr[\cE_{xy'}] = \Omega\left(\frac{\eps}{\sqrt{n}}\right)\cdot  \Pr[\cE_{xy}].$$
\end{proof}
%Since there is a one-to-one mapping from $\cE_{xy}$  to $\cE_{xy'}$ (this is where the fact that $E$ is a matching is used), we get that 

Going back to \Eqn{mat}, 
\begin{align*}
\Pr[\textrm{ success }] & \geq \sum_{(x,y) \in \Pi} \Pr[\cE_{x,y}]\cdot \frac{\Pr[\cE_{x,E(y)}]}{\Pr[\cE_{x,y}]} \\
& = \Omega\Big(\frac{\sigma}{\sqrt{n}} \cdot \sum_{(x,y) \in \Pi} \Pr[\cE_{x,y}]\Big)
\end{align*}
Now for the punchline. Color all points in $B_1$ blue. 
By the choice of parameters, $|B_1| = \sigma 2^n$.
\Lem{blue} tells us that the probability of \pathtester$(\sigma)$ sampling a pair $(x,y)$ such that both points are blue is $\Omega(\sigma^2/\ln(1/\eps))$.
This is the event that $x \prec y$ (or $y \prec x)$, $\|y - x\|_1 \geq \frac{\sigma\ell}{32C_\eps} - 1$,
and $x,y \in B_1$. The probability of this event is exactly twice $\sum_{(x,y) \in \Pi} \Pr[\cE_{x,y}]$.
Hence, the probability of success is 
%$$\Pr[\textrm{ success }] = 
$\Omega\left(\frac{\sigma^3}{\sqrt{n}\ln(1/\eps)}\right)$.
%\noindent
%Therefore the probability of success is at least
%\begin{eqnarray*}
%\sum_{(x,y')\in B_1\times B_0} \Pr[\cE_{xy'}] & = &  \Omega\left(\frac{\eps}{\sqrt{n}}\right)\cdot \sum_{(x,y)\in B_1\times B_1} \Pr[\cE_{xy}] \\
%& = & \tomega\left(\frac{\eps \eta^3}{\sqrt{n}}\right)
%\end{eqnarray*}
%where the last inequality follows from \Eqn{5}.   This completes the proof of the lemma.
\end{proof}

\subsection{Wrapping it up with the dichotomy} \label{sec:dich}

We state the directed variant of Margulis' theorem. We actually prove a slightly stronger 
dichotomy theorem between the total number of violated edges and largest matching of violated edges.
We let $\Phi^+_f$ be the number of violated edges divided by $2^{n-1}$ (think of this as the ``violation
influence"). We set $\Gamma^+_f$ to be the size of the largest matching of violated edges
\emph{in the middle layers} divided by $2^n$.

%
%\begin{theorem} \label{thm:dich}For any function $f$ that is $\eps$-far from monotone,
%and for any $s> 0$, either 
%there is a matching of violated edges completely contained in the middle layers of  cardinality at least $\eps2^{n-5}/s$, or the {\em total} number of violated edges is at least $\eps s 2^{n-1}$.
%\end{theorem}

\begin{theorem} \label{thm:dich}For any function $f$ that is $\eps$-far from monotone,
$\Phi^+_f \cdot \Gamma^+_f \geq \frac{\eps^2}{32}$.
\end{theorem}

\begin{proof}
Recall, since $f$ is $\eps$-far from monotonicity, any maximal matching $M$ of violated {\em pairs} (not edges) must have cardinality $|M|\geq \eps2^{n-1}$.
For any such matching $M$ of violated pairs, define the \emph{average length of $M$} to be the quantity
$$|M|^{-1}\sum_{(x,y)\in M} \|y - x\|_1$$
Choose $M$ to be a  \emph{maximum} cardinality matching of violated pairs with the smallest average length, denoted by $r$.
%Since $f$ is $\eps$-far from being monotone, and $M$ is a maximal matching, $|M| \geq \eps 2^{n-1}$.
\begin{lemma}\label{lem:piece3}
If the average length of $M$ is $r$, then $\Gamma^+_f \geq \frac{\eps}{32r}$. \\That is,  
there exists a matching $E$ of violated edges all lying in the middle layers of the hypercube with size at least $\frac{\eps 2^{n}}{32r}$.
\end{lemma}
\begin{proof}Deferred to \Sec{p3}.\end{proof}
\begin{lemma}\label{lem:piece4}
If the average length of $M$ is $r$, then $\Phi^+_f \geq r\eps$. \\ 
That is, 
 there are at least $r \eps 2^{n-1}$ violated edges.
\end{lemma}
\begin{proof}Deferred to \Sec{p4}.\end{proof}

The proof of the theorem follows from the above two lemmas.
\end{proof}

 %This holds directly from the previous two lemmas. 
%By \Lem{piece3}, the largest matching of violated edges in the middle layers has size at least $\eps 2^n/48r^2$, and implies the theorem if $r\leq s$.
%If (for the choice of $s$) the largest matching has size at most $\eps 2^n/48s^2$, then $r \geq s$.
%Then by \Lem{piece4}, 
%there are at least $\eps r 2^{n-1}\geq s\eps2^{n-1}$ violated edges.

It is now routine to prove \Thm{main} and \Thm{asmon}.

\begin{proofof}{\Thm{main}} Set $s = n^{1/8} \eps^{3/2}$ and $\sigma = n^{-1/8} \eps^{1/2}/32$.
%for some sufficiently large constant $c$.
We will argue that either the edge tester or \pathtester$(\sigma)$ has a success probability
of $\Omega(n^{-7/8}\eps^{3/2}\ln^{-1}(1/\eps))$.

The success probability of the edge tester is exactly $\Phi^+_f/n$. If $\Phi^+_f \geq s$,
then the edge tester succeeds with the desired probability. So let us assume that $\Phi^+_f < s$.
By \Thm{dich}, $\Gamma^+_f  \geq \frac{\eps^2}{32s} = \sigma$.
%(We set $c$ large enough to ensure that $\sigma$ is smaller than this bound.)
Therefore, 
there exists a matching of violated edges of size $\sigma 2^n$, all of them lying in the middle layers.
We apply \Lem{piece2} and bound the success probability of \pathtester$(\sigma)$
by $\Omega\left(\frac{\sigma^3}{\sqrt{n}\ln(1/\eps)}\right) = \Omega\left(n^{-7/8}\eps^{3/2}\ln^{-1}(1/\eps)\right)$.
%
%
%Suppose the largest matching in the middle layers has size at least $\eps 2^{n-5}/s$.
%By \Lem{piece2}, we get that the path tester succeeds with probability
%$$\tomega\left[\frac{\eps}{\sqrt{n}}\cdot \Big(\frac{\eps}{s}\Big)^{2}\right] = \tomega\left(\frac{\eps^{3}}{s^2\sqrt{n}}\right) = \tomega(n^{-5/6}\eps^{5/3})$$
%If the largest matching is at most $\eps 2^{n-5}/s$, then by \Thm{dich} there are at least $\eps s 2^{n-1}$ violated
%edges. The edge tester succeeds with probability $\Omega(s\eps/n) = \Omega(n^{-5/6} \eps^{5/3})$. 
\end{proofof}
\as{
\begin{proofof}{\Thm{asmon}} Note that $\Phi^+_f \leq \I(f)$. From \Thm{dich}, 
$\Gamma^+_f \geq \frac{\eps^2}{32\I(f)}$. If we set $\sigma$ to be this lower bound,
then \pathtester$(\sigma)$ succeeds with probability $\Omega(n^{-1/2}\eps^6\I^{-3}(f)\ln^{-1}(1/\eps))$.
%
%with $s = \I(f)/\eps$ in \Thm{dich}, we know there exists a 
%violated edge matching of cardinality $\Omega(\eps^2 2^n / \I(f))$ in the middle layers.
%By \Lem{piece2}, the path tester succeeds with probability $\tomega(n^{-1/2} \eps^5/\I^2(f)))$. 
We do not need the edge tester for these functions.
%
%Let $M$ be the maximum cardinality matching which minimizes the average length, as in Piece 4, and let $r$ be the average length. \Lem{piece4} gives us that the number of violated edges is at least $\eps r 2^{n-1}$. This in turn implies that $\I(f)\geq \eps r/2$; every violated edge is a `sensitive' edge. In other words, $r = O(\I(f)/\eps)$. \Thm{asmon} now follows from \Lem{piece3} and \Lem{piece2}.
\end{proofof}
}

\subsection{Proof of \Lem{piece3}}\label{sec:p3} 
\noindent

We first state the routing theorem of Lehman and Ron. Recall that $L_i$ was defined as the set of points of the hypercube with exactly $i$ ones.
\begin{theorem}[Lehman-Ron~\cite{LR01}] \label{thm:LR}
Let $S \subseteq L_i$ and $R\subseteq L_j$ with  $|S| = |R| = m$ and $i < j$. Furthermore, suppose there is a bijection $\phi:S \mapsto R$ such that $x \succ \phi(x)$, $\forall x \in S$, that is, $(x,\phi(x))$ are ancestor-descendants. Then there exists $m$ {\em vertex disjoint directed paths} from the set $S$ to $R$.
\end{theorem}
For convenience, we represent the matching through ordered pairs, so if $(x,y) \in M$, then $x \prec y$ (and $f(x) = 1$, $f(y) = 0$).
%We define an \emph{uncrossing} and \emph{straightening} operation on violation matchings. 
Recall $M$ is a maximum cardinality matching in the violated graph with the smallest average length. Among such matchings, let $M$ actually be one maximizing 
$$\Psi(M) := \sum_{(x,y)\in M} ||x-y||^2_1$$

We prove a structural claim regarding $M$. Two pairs $(x,y)$ and $(x',y')$ {\em cross} if (a) 
there exists a $z$ such that 
$x \prec z \prec y$ and $x' \prec z \prec y'$ and (b) the intervals $[|x|,|y|]$ and 
$[|x'|,|y'|]$ strictly cross, meaning that neither interval contains the other.
Note that the interval $[|x|,|y|]$ and $[|x'|,|y'|]$ are ranges of integers.
(By (a), the intervals $[|x|,|y|]$ and 
$[|x'|,|y'|]$ must intersect.) Recall $|x|$ is the number of ones in $x$.
%
%
%, that is, either $|x'|$ lies in $(|x|,|y|)$ or $|x|$ lies in $(|x'|,|y'|)$, and (b) .

\begin{claim}\label{clm:ucross}
There are no crossing pairs in $M$.
\end{claim}
\begin{proof}
Suppose $(x,y)$ and $(x',y')$ cross. Consider $M'$ formed by deleting these pairs from $M$ and adding $(x,y')$ and $(x',y)$. These are valid violations due to the presence of the vertex $z$. Furthermore, observe that $\Psi(M') > \Psi(M)$ since the sum of squares of a pair of numbers having a fixed sum increases as the maximum (of the pair) increases.
\end{proof}

%We apply \Clm{uncross}
%to $M'$ to get a minimal violation matching $\widehat{M}$.

For every two levels $i < j$ of the hypercube, let $M_{i,j} \subseteq M$ be the subset of pairs with endpoints in the level sets $L_i$ and $L_j$.
Apply \Thm{LR} to get a collection of $|M_{i,j}|$ vertex disjoint paths. Each of these vertex disjoint paths contain at least one violated edge,
and let $F_{i,j}$ be the set of these edges. Note that $F_{i,j}$ forms a matching. Consider the {\em multiset} $F$ formed by the union of $F_{i,j}$ over the set $\{(i,j): i,j\in I_\ell,~i < j, ~j - i \leq 2r\}$. Note that $F$ may contain more than one copy of the same edge. Also note that all edges of $F$ lie in the middle layers.

\begin{claim}\label{clm:size}
$|F| \geq |M|/4$.
\end{claim}
\begin{proof}
Note that $|F| = \sum_{(i,j): i,j \in I_\ell, i < j, j-i\leq 2r} |M_{i,j}|$.
Since the matching $M$ has average length $r$, by Markov's inequality, at least 
$|M|/2$ of these  pairs have length at most $2r$. Furthermore, from \Prop{hyper} we get that at most $\eps^52^n \leq |M|/4$ pairs in $M$ have endpoints not in the middle layer. Looking at the remainder, we get $\sum_{(i,j): i,j \in I_\ell, i < j, j-i\leq 2r} |M_{i,j}| \geq |M|/4$.
\end{proof}

\begin{claim}\label{clm:degree}
No point $z\in \{0,1\}^n$ has more than $2r$ edges of $F$ incident on it.
\end{claim}
\begin{proof}
Pick a vertex $z$, and pick any two edges $f_1$ and $f_2$ of $F$ incident on it. Since each $F_{i,j}$ is a matching, these must lie in different $F_{i,j}$'s. Suppose they are $F_{i,j}$ and $F_{a,b}$, where $i$ could be $a$ and $j$ could be $b$, but not both together.
Note that $i \leq |z| \leq j$ and $a \leq |z| \leq b$.
 We claim that  $[i,j]$ and $[a,b]$ cannot cross, and therefore one must strictly lie in the other. 
 There can be at most $2r$ intervals containing $|z|$ satisfying such containment relationships.
% 
% Therefore the set of (level) intervals corresponding to various edges of $F$ incident on $z$ form a chain and have a maximum length of $2r$.
Thus, there can be most $2r$ edges of $F$ incident on $z$.
 
 We claim that $[i,j]$ and $[a,b]$ cannot cross. 
 %Construct a matching $M'$ with the same cardinality and average length as $M$ but $\Phi(M) < \Phi(M')$.
 To see this,  consider the pairs in $M_{i,j}$, and let them be $(x_1, y_1), (x_2, y_2), \ldots, (x_k, y_k)$.
Note that  \Thm{LR} implies $k$ vertex disjoint paths containing all these vertices. Hence, there is some
permutation $\pi$ such that for each $i \in [k]$, there is a path from $x_i$ to $y_{\pi(i)}$. 
Let $M'_{i,j} = \{(x_i,y_{\pi(i)})\}$. Similarly, define $M'_{a,b}$.
Let $M'$ be the matching where $M_{i,j}$ and $M_{a,b}$ are replaced by $M'_{i,j}$ and $M'_{a,b}$, respectively, and all other pairs remain.
Note that $M'$ has the same average length, same cardinality {\em and} $\Phi(M') = \Phi(M)$. 
But now, we have a pair $(x,y) \in M'_{i,j}$ such that $x\prec z \prec y$, and a pair $(x',y') \in M'_{a,b}$ such that $x'\prec z\prec y'$.
This is because $z$ is incident to an edge in both $F_{i,j}$ and $F_{a,b}$. If $[i,j]$ and $[a,b]$ cross, then $(x,y)$ and $(x',y')$ cross.
Since $M'$ maximizes the potential $\Phi$, \Clm{ucross} is contradicted.
% which by the proof of \Clm{ucross} will give a matching $M''$ with the same average length and cardinality with $\Phi(M'') > \Phi(M') = \Phi(M)$. 
\end{proof}

Since the multigraph induced by $F$ has maximum degree $2r$, there exists a matching $E\subseteq F$ of size $|E| \geq |F|/4r$. This can be obtained by picking an edge in $E$ arbitrarily and deleting all edges incident to any of its endpoints from $F$. For every edge added to $E$ there are at most $4r$ edges deleted from $F$. Therefore, following \Clm{size}, we get a matching $E$ of violated edges of size $\geq |M|/16r \geq \eps2^{n-5}/r$. Furthermore, they all lie in the middle layers. This completes the proof of \Lem{piece3}.

\subsection{Proof of \Lem{piece4}}\label{sec:p4}
\noindent

Let $M_i$ be the set of pairs in $M$ that cross dimension $i$, that is, $M_i := \{(x,y)\in M: x_i=1, y_i=0\}$. 
\Lem{piece4} follows from following theorem. 

%\submit{\vspace{-3mm}}
\begin{theorem}\label{thm:cs}
For all $i$, the number of violated edges across dimension $i$ is at least $|M_i|$.
\end{theorem}
%\submit{\vspace{-5mm}}

Note that 
$\sum_i |M_i| = \sum_{(x,y)\in M}\|y - x\|_1 = r|M|$, since a pair $(x,y)$ appears in precisely $\|y-x\|_1$ different $M_i$'s. \Thm{cs} implies that the number of violated edges is at least $\sum_i |M_i|$. \Lem{piece4} follows because $|M|\geq \eps2^{n-1}$. 

We now prove \Thm{cs}.

\begin{proof}%(Proof of \Thm{cs}.)
This requires setting up some of the machinery of~\cite{CS12}. 
Let $H$ be the {\em perfect} matching of the hypercube formed by the {\em edges} crossing the  $i$th dimension. 
%Note that $H$ is a perfect matching. 
%We look at the alternating paths and cycles formed by the symmetric difference of $M$ and $H_i$. 
Let $X$ be the endpoints of $M_i$. For all $x\in X$, we now define a sequence $\bS_x$.
In what follows, we use the shorthand $M(v)$ and $H(v)$ to denote the partners of $v$ in the matchings $M$ and $H$, respectively.
The first term of the sequence, $\bS_x(0)$, is $x$.
For even $i$, $\bS_x(i+1) = H(\bS_x(i))$.
For odd $i$, if $\bS_x(i)\in X$, or is $M$-unmatched, then $\bS_x$ terminates.  Otherwise, $\bS_x(i+1) = M(\bS_x(i))$.  
The best way to think about $\bS_x$ is via alternating paths and cycles formed by the matchings $M$ and $H$.
We start at $x$ and take the $H$-edge along the alternating path. We alternate between $H$ and $M$ edges  till we reach an endpoint of the alternating path or another vertex in $X$. Thus, each $\bS_x$ terminates.
It is not hard to see that if $\bS_x$ ends at $y\in X$, then $\bS_y$ is just $\bS_x$ in reverse. Also note that $\bS_x$ and $\bS_y$ are disjoint unless $y$ terminates $\bS_x$, for otherwise $\bS_x$ would've terminated earlier. Therefore the number of sequences is at least $|X|/2 = |M_i|$.
The theorem is a consequence of the following lemma.

\begin{lemma} \label{lem:bs} For all $x$, $\bS_x$ contains a violated edge in $H$. 
\end{lemma}

\begin{proof} We will prove this through contradiction, and will henceforth assume that (for some $x$), $\bS_x$ has no violated edge in $H$.
We show that $\bS_x$ cannot terminate, completing the contradiction.
For brevity, let us use $s_j$ to denote $\bS_x(j)$. Let $(x,y)$ be the pair in $M_i$.
We use $s_{-1}$ to denote $y$. Wlog, assume $x \succ y$, thus $x_i=1$ and $y_i=0$. Also $f(y)=1$ and $f(x)=0$ since the pair is a violation. 

Let $\D_b$ ($b = \{0,1\}$) be the $n-1$ dimensional hypercube where $i$th coordinate is $b$.
We will use $d(x,x')$ for the Hamming distance between two points $x$ and $x'$.
We have the following simple claim.

\begin{claim} \label{clm:j} Let $j \geq 0$ be an index and suppose $s_j$ exists.

\begin{minipage}[t]{0.25\linewidth}
If $j \equiv 0 {\md 4}$, 
\begin{asparaitem}
	\item $f(s_j) = 0$.
	\item $s_j \in \D_1$.
\end{asparaitem}
\end{minipage}
\begin{minipage}[t]{0.25\linewidth}
If $j \equiv 1 {\md 4}$,
\begin{asparaitem}
	\item $f(s_j) = 0$.
	\item $s_j \in \D_0$.
\end{asparaitem}
\end{minipage}
\begin{minipage}[t]{0.25\linewidth}
If $j \equiv 2 {\md 4}$,
\begin{asparaitem}
	\item $f(s_j) = 1$.
	\item $s_j \in \D_0$.
\end{asparaitem}
\end{minipage}
\begin{minipage}[t]{0.25\linewidth}
If $j \equiv 3 {\md 4}$,
\begin{asparaitem}
	\item $f(s_j) = 1$.
	\item $s_j \in \D_1$.
\end{asparaitem}
\end{minipage}
%
%
%For $j \equiv 0 {\md 4}$, $f(s_j) = 0$ and $s_j \in \D_1$;
%$j \equiv 1 {\md 4}$, $f(s_j) = 0$ and $s_j \in \D_0$;
%$j \equiv 2 {\md 4}$, $f(s_j) = 1$ and $s_j \in \D_0$;
%$j \equiv 3 {\md 4}$, $f(s_j) = 1$ and $s_j \in \D_1$.
\end{claim}

\begin{proof} We prove by induction on $j$. For the base case, $s_0 = x$, and $f(x) = 0$, $s_0 \in \D_1$.
Consider $j \equiv 0 {\md 4}$, $j \geq 1$. By the induction hypothesis, $f(s_{j-1}) = 1$ and $s_{j-1} \in \D_1$.
Since $s_j = M(s_{j-1})$ and thus $(s_{j-1},s_j)$ is a violation, $f(s_j) = 0$ and $s_j \in \D_1$.
Consider $j \equiv 1 {\md 4}$. By the induction hypothesis, $f(s_{j-1}) = 0$ and $s_{j-1} \in \D_1$.
Since $s_j = H(s_{j-1})$ and $(s_{j-1},s_j)$ is not a violation by assumption that no $H$-edge is a violation, $f(s_j) = 0$ and $s_j \in \D_0$.
The remaining cases are analogous.
\end{proof}

\begin{claim} \label{clm:dist} Let $j \geq 0$ be even. Then $(s_j, s_{j+3})$ is a violation
and $d(s_j, s_{j+3}) = d(s_{j+1}, s_{j+2})$. Also, the pair $(y,s_1)$ is a violation
and $d(y,s_1) = d(y,s_0) - 1$.
\end{claim}

\begin{proof} Suppose $j \equiv 2 {\md 4}$. From \Clm{j}, $f(s_j) = 1$ and $f(s_{j+3})=0$.
The claim also implies $f(s_{j+1})=1$ and since $(s_{j+1},s_{j+2}) \in M$, $s_{j+2} \succ s_{j+1}$. Furthermore,
both $s_{j+1},s_{j+2}\in \D_1$.
Since $s_j = H(s_{j+1})$ and
$s_{j+3} = H(s_{j+2})$, we get (a) $s_{j+3} \succ s_j$ as well, implying $(s_j,s_{j+3})$ is a violation, and
(b) $d(s_j,s_{j+3})=d(s_{j+1},s_{j+2})$.  The case $j \equiv 0 {\md 4}$ is analogous.

%
%$s_j \in \D_0$
%and $f(s_{j+1})=1$ and $s_{j+1}\in \D_1$.
%Observe that $s_{j+2} = M(s_{j+1})$, and $s_{j+3} = H(s_{j+2})$.
%%Since no $H$-edge is a violation, we get $f(s_{j+1})=1$ as well.
%Since $f(s_{j+1}) = 1$ and $(s_{j+1},s_{j+2}) \in M$, $s_{j+2} \succ s_{j+1}$
%and $s_{j+2} \in \D_1$. Now, $f(s_{j+3}) = 0$ and $s_{j+3} = H(s_{j+2})$.
%Finally, $f(s_{j+3}) = 0$ and $s_j \prec s_{j+3}$. So $(s_j, s_{j+3})$ is a violation
%and $d(s_j, s_{j+3}) = d(s_{j+1}, s_{j+2})$.

Observe that $(y,s_0)$ is a violation where $y \in \D_0$ and $s_0 \in \D_1$. Since $s_1 = H(s_0)$,
$s_1$ has the same coordinates as $s_0$ except for the $i$th one. Also, $f(s_1) = 0$.
Therefore, $(y,s_1)$ is a violation and $d(y,s_1) = d(y,s_0) - 1$.
\end{proof}

We argue that $\bS_x$ cannot terminate, by showing for any odd $j$, $s_{j+1}$ must exists.
(This is trivially true for even $j$, since $H$ is a perfect matching.)
We focus on $j \equiv 1 {\md 4}$ (the case $j \equiv 3 \md 4$ is analogous).
Because $f(s_j) = 0$ and $s_j \in \D_0$, $s_j$ cannot participate in a pair in $M_i$.
Hence, $s_j \notin X$. Suppose $s_j$ was unmatched. Consider the following set
of pairs in $M$: $A = \{ (s_k, s_{k+1}) | \textrm{$k$ odd}, -1 \leq k \leq j-2\}$.
Suppose we replaced these pairs in $M$ by $B = \{(y,s_1)\} \cup \{(s_k,s_{k+3}) | \textrm{$k$ even}, 0 \leq k \leq j-3\}$.
Note that $|A| = \ceil{j/2} = |B|$. Also, by \Clm{dist}
$$ d(y,s_0) + \sum_{\substack{k = 1 \\ \textrm{$k$ odd}}}^{j-2} d(s_k,s_{k+1}) = d(y,s_1) -  1 + \sum_{\substack{k = 0 \\ \textrm{$k$ even}}}^{j-3} d(s_k,s_{k+3}) $$
This means that replacing $A$ and inserting $B$ (in $M$) leads to a violation matching of the same size
with a smaller Hamming distance. This violates the property that $M$  has minimum average Hamming distance, and therefore, the sequence $\bS_x$ cannot terminate. This cannot occur, and therefore, every $\bS_x$ must contain a violated edge in $H$. This ends the proof of \Lem{bs}.
%
%
%Note that $s_1 = H(x)$ has $i$th coordinate $0$. Since there is no violating edge, $f(s_1) =0$ as well. Note that $s_1\succ y$, as well, and therefore $(y,s_1)$ forms a violating pair. If $s_1$ were not matched in $M$, then $M-(x,y)+(y,s_1)$ would be a matching which would decrease the average length. Therefore $s_1$ must be matched. Since $f(s_1)=0$, $M(s_1) = s_2 \prec s_1$. In particular, the $i$th coordinate of $s_2$ and $s_1$ are both $0$, and therefore the edge does not lie in $M_i$. So $s_2\notin X$. The same argument shows no other vertex can be in $X$.
%Once again, check that $s_3 = H(s_2) \prec s_0 = x$. If $s_3$ were not matched, then we can replace the edges $(x,y)$ and $(s_1,s_2)$ in $M$ by $(y,s_1)$ and $(s_0,s_3)$ to get a matching with smaller average length. Therefore, anytime we reach a vertex which is not $M$-matched, we get a better matching. Therefore, we never encounter a vertex which is in $X$ or is $M$-unmatched. This contradicts the termination of $\bS_x$, and therefore, there must exist a violating edge in $\bS_x$. This proves the theorem.
\end{proof}\end{proof}

\section{Conclusion}
In this paper, we make progress on the question of testing monotonicity of Boolean functions over the hypercube. 
%In particular, we exhibit a $o(n)$-query tester. 
%\Thm{dich}, \Lem{piece2} and \Lem{blue} are tight, and the exponent of $5/6$ is the best we can hope for using our analysis. 
Our approach falls short of the known $\Omega(\sqrt{n})$ lower bound for one-sided, non-adaptive testers.
Nevertheless, we believe the path tester (alone) is a $O(\sqrt{n})$-query monotonicity tester for Boolean functions. 
A possible approach is suggested by \Thm{asmon}. Can we perform a different analysis (or even design a different algorithm) for high average sensitivity
functions?

\section{Acknowledgements}

We thank Cyrus Rashtchian and Naomi Kirshner for pointing out a mistake in an older version.

Sandia National Laboratories is a multi-program laboratory managed and operated by Sandia Corporation, a wholly owned subsidiary of Lockheed Martin Corporation, for the U.S. Department of Energy's National Nuclear Security Administration under contract DE-AC04-94AL85000. CS is grateful for the support received from the Early Career LDRD program at Sandia National Laboratories.

%Although our current exponent is far from the known lower bound of $\sqrt{n}$
%, for a large class of functions we show that our analysis suffices to give optimal results.

%There is some scope of improving the dependence $n,\eps$ in the previous analysis. 
%How far can our current analysis hope to bring down the exponent? 
%The exponent of $\eps$ in \Lem{piece1} cannot be smaller than $2$, and the exponent of $r$ in \Lem{piece3} cannot be smaller than $1$. With similar calculation as in our proof, one could get the exponent of $n$ down to $5/6$, but no smaller. 
%So 
%We have not been able to do so, as yet, but our current argument cannot come much closer to the $\sqrt{n}$ lower bound. 
%Nevertheless, we believe the path tester (alone) is a $O(\sqrt{n})$-query monotonicity tester for Boolean functions. But as is with many problems in property testing, the crux lies in the analysis. We would probably need a new idea than those described in this current paper to make this go through. This is the obvious open direction. 
%
%Finally, there may be a different algorithm that one may think of for `high' average sensitivity functions. 
\bibliographystyle{amsplain}
\bibliography{derivative-testing}

\end{document}